\documentclass[10pt, journal, letterpaper]{IEEEtran}
\usepackage[utf8]{inputenc}
\usepackage{graphicx}
\usepackage{listings}
\usepackage{amsmath}
\usepackage{amssymb}
\usepackage{stackrel}
\usepackage{arydshln}
\usepackage{framed}
\usepackage{caption}
\usepackage{subcaption}

\usepackage{algorithm}
\usepackage{cite}

\usepackage{amsthm}

\newtheorem{definition}{Definition}
\newtheorem{claim}{Claim}

\usepackage{algorithmic}
\usepackage{hyperref}
\usepackage[shortlabels]{enumitem}
\usepackage{xcolor}
\newcommand{\E}{\mathbb{E}}

\renewcommand{\P}{\mathbb{P}}

\newcommand{\eps}{\epsilon}

\begin{document}
\title{Centralized Versus Distributed Routing for Large-Scale Satellite Networks}
\author{Rudrapatna Vallabh Ramakanth and Eytan Modiano \thanks{R.V. Ramakanth and E. Modiano are affiliated to the Laboratory of Information and Decision Systems (LIDS) at the Massachusetts Institute of Technology (MIT), Cambridge, MA 02139, USA. Email: \{vallabhr, modiano\}@mit.edu.}}

\maketitle

\begin{abstract}
An important choice in the design of satellite networks is whether the routing decisions are made in a distributed manner onboard the satellite, or centrally on a ground-based controller. We study the tradeoff between centralized and distributed routing in large-scale satellite networks. In particular, we consider a centralized routing scheme that has access to global but delayed network state information and a distributed routing scheme that has access to local but real-time network state information. For both routing schemes, we analyze the throughput and delay performance of shortest-path algorithms in networks with and without buffers onboard the satellites. We show that distributed routing outperforms centralized routing when the rate of changes in the network link state is comparable to the inherent propagation and transmission delays. In particular, we show that in highly dynamic networks without buffers, the distributed scheme achieves higher throughput than a centralized scheme. In networks with buffers, the distributed scheme achieves lower delays.
\end{abstract}

\section{Introduction}
Recently, there has been a growing interest in deploying large-scale satellite networks for data communication. These networks usually comprise Low-Earth Orbit (LEO) constellations or Medium-Earth Orbit (MEO) constellations that are designed to be dense and regular in structure, resembling a grid-like structure. 
For data communication purposes, designing a good routing strategy for such large-scale networks is a challenging task. The underlying routing scheme of a network plays a critical role in its performance. In the context of satellite networks, an important design choice is whether routing decisions should be made on-board the satellites, or on a ground-based central controller. 

It is not surprising that most satellite operators today route packets through their networks using centralized routing approaches, as these methods provide tighter control and troubleshooting from the ground \cite{fraire_centralized_2021}. In addition, a centralized algorithm is easier to implement and requires less resources onboard satellites. However, in such a paradigm, the controller cannot react to instantaneous changes in the network state, such as congestion effects or link failures. This is because there is an innate unavoidable delay in propagating, aggregating, and processing network state information from satellites to the ground (and vice versa), which ranges in the order of tens of milliseconds or longer. 

On the other hand, there is also a strong appeal to use distributed routing for satellite networks. A distributed routing strategy enables satellites to react to instantaneous network changes in their neighborhood while making routing decisions. Distributed routing policies for satellite networks that exploit the density and structure of these networks have been proposed \cite{ekici_distributed_2001, page_distributed_2023, liu_low-complexity_2015, taleb_sat04-3_2006, deng_distance-based_2023}. 
However, such a scheme would incur additional infrastructure costs for route computation and buffering onboard the satellite. Moreover, even though it is conceptually possible for satellites to obtain global state information for routing, it is impractical because of limited processing and storage capabilities onboard the satellites, along with the additional overhead required to transmit large amounts of network state information. This forces distributed algorithms to only use local information for routing.
Thus, a trade-off arises between centralized and distributed routing strategies. 

To that end, we seek to answer the following question: \textbf{under what network conditions does a distributed routing scheme outperform a centralized routing scheme in terms of throughput and delay in regular, yet dynamic networks?}

Routing is a well-studied problem. In most terrestrial networks, operators use distance-vector routing protocols or link-state routing protocols, both of which are relatively static shortest-path routing protocols \cite{bertsekas_data_1992}. These are not suitable for satellite networks, where the links change state on the same time scale as the transmission of a packet across the network \cite{ekici_distributed_2001}. There exist methods to compute optimal routing policies using flow-based models \cite{bertsekas_data_1992}. However, these policies are very sensitive to traffic conditions and network parameters, rendering them impractical \cite{bertsekas_data_1992}. On the other hand, while routing for dynamic \textit{ad hoc} wireless networks has been studied for a long time \cite{corson_distributed_1995, perkins_ad-hoc_1999, johnson_dynamic_1996, karp_gpsr_2000}, they are not suitable for satellite networks, as they would disregard much of the regular structure. Of course, there exists a plethora of prior work on studying various routing strategies for satellite networks \cite{modiano_efficient_1996,werner_dynamic_1997, ekici_distributed_2001, neely_power_alloc_2003, jun_sun_capacity_2002, sun_routing_2004, taleb_sat04-3_2006, liu_low-complexity_2015, fraire_centralized_2021, wang_orbit-grid-based_2023, page_distributed_2023, deng_distance-based_2023, li_stable_2024, zheng_sdn_2024}. 
However, to the best of our knowledge, existing work does not analytically characterize the conditions under which a distributed (or centralized) routing strategy yields better performance.

In this work, we consider routing packets across a grid-like network topology (see Fig. \ref{fig:network-model}) with dynamic links that evolve in time according to a Markov process. Finding an optimal route in such a network requires the solution to a dynamic program (Markov Decision Process) with a state space that grows exponentially with the network size \cite{akrida_how_2020_journal}. Since the solution to such a dynamic program is computationally intractable, we focus on practical and representative shortest-path routing strategies. In particular, for the centralized case, we consider Shortest Connected Path Routing policy that selects the shortest path from the source to the destination \textit{only using links that are observed to be available for use}. For the distributed case, we consider a Greedy Routing policy, that routes packets only along links that take it closer to the destination at every step. Note that both algorithms are types of ``shortest-path algorithms" and are loop-free.

The main contribution of this work is in characterizing when distributed routing outperforms centralized routing. In particular, when the rate of change of link states is comparable to the inherent propagation and transmission rates of the network, we show that in networks without onboard satellite buffering, distributed routing generally achieves a higher throughput, and in networks with onboard buffering, distributed routing achieves a lower delay. However, in a network where rate of change of link states is gradual, the centralized scheme has lower delay and higher throughput in networks with and without buffers respectively.

The rest of the paper is organized as follows. In Section \ref{sec:system_model}, we describe the model of the satellite network, definitions of the routing schemes and the performance metrics in detail. Section \ref{sec:centralized_routing_analysis} analyzes the performance of a practical centralized routing algorithm. Similarly, Section \ref{sec:distributed_routing_analysis} analyzes the performance of a practical distributed routing algorithm. In Section \ref{sec:comparison}, we compare the results of the analysis in the former sections to understand the conditions favorable for centralized routing and decentralized routing. Finally, we conclude the paper in Section \ref{sec:conclusion}.

\section{System Model and Preliminaries}\label{sec:system_model}

A satellite network has a lot of dynamic connections. The topology of the network evolves in the order of minutes, whereas routing over the network happens in the order of milliseconds. For this reason, the satellite network is modeled as a collection of snapshots with a fixed topology \cite{werner_dynamic_1997}. We are interested in the routing of packets in a given snapshot of the satellite network. We will now describe our model for the satellite network.

\subsection{Satellite Network Model}
\begin{figure}
\centering\includegraphics[width=0.3\textwidth]{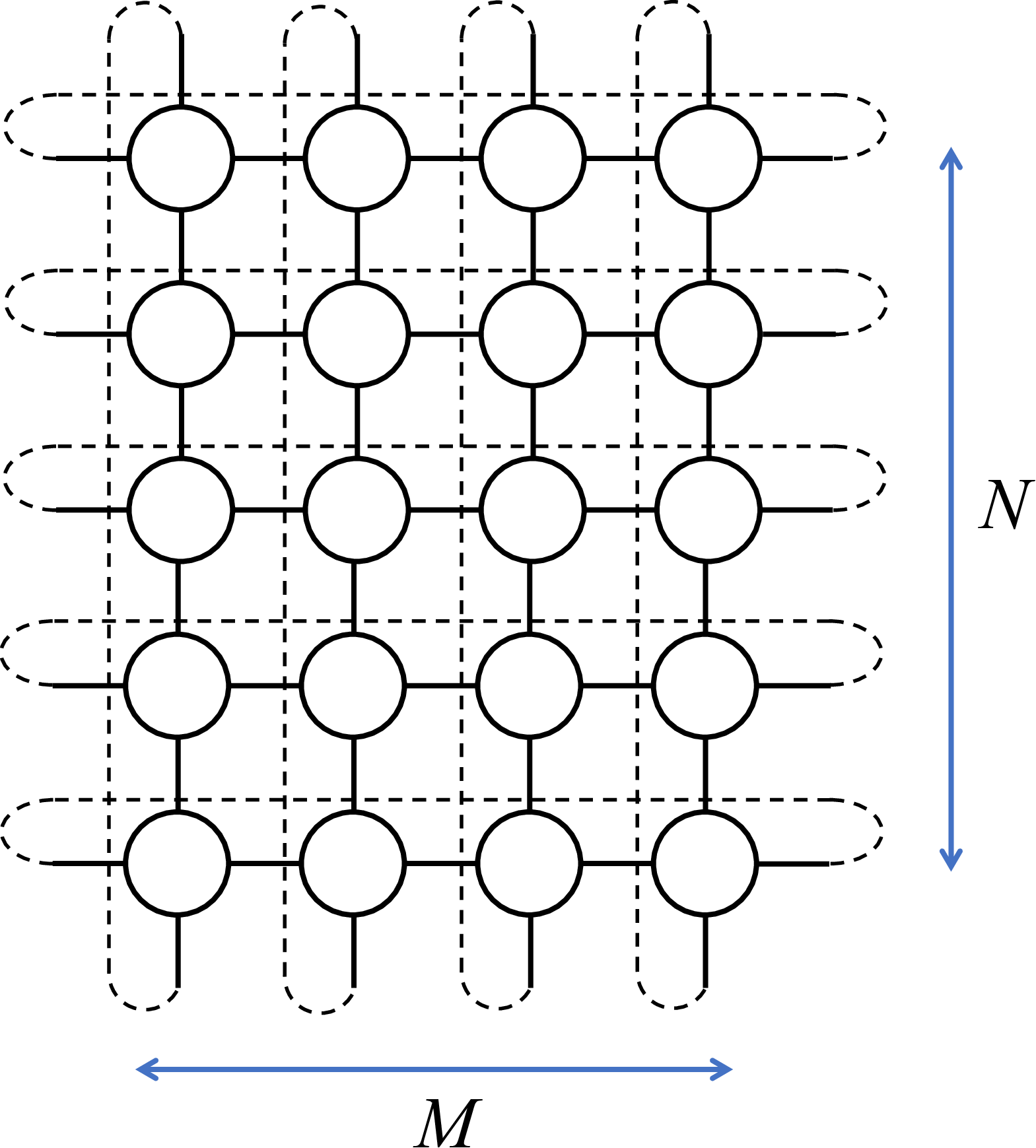}
    \caption{Toroidal Mesh grid of size $N \times M$, where $N = 5$ and $M = 4$. Each node has degree 4.}
    \label{fig:network-model}
\end{figure}
Consider a toroidal mesh grid of size $N \times M$ (see Figure \ref{fig:network-model}). This models a particular snapshot of the satellite network. It is an abstract logical network model in which every satellite node is connected to its four nearest neighbors-- two links in the same orbital plane and two links which are cross-orbit \cite{sun_routing_2004, jun_sun_capacity_2002, handley_delay_2018}. Here, $N$ denotes the number of satellites per orbital plane and $M$ denotes the number of orbital planes. 
In this model, there are $NM$ nodes. Each node is labeled with a coordinate $(x,y)$, such that 
\begin{align}
    x &\in \left\{\left\lfloor-\frac{M}{2}\right\rfloor+1, ..., -1, 0, 1, ..., \left\lfloor \frac{M}{2} \right\rfloor\right\}, \text{ and }\\
    y &\in \left\{\left\lfloor-\frac{N}{2}\right\rfloor+1, ..., -1, 0, 1, ..., \left\lfloor \frac{N}{2} \right\rfloor\right\}.
\end{align}
 We divide the duration of the snapshots into discrete time slots, where each time slot is the time it takes to transmit a fixed-size packet from one satellite node to its neighbor. For simplicity, we assume that the time taken by a fixed-length packet to be transmitted across any link in this network is the same and equal to one time slot. Henceforth, we measure time in the units of slots. We consider a directed network, where the link from node $(x_1,y_1)$ to $(x_2, y_2)$ is different from the link from $(x_2,y_2)$ to $(x_1, y_1)$.
The links in this network are modeled to be random, with dynamics that evolve according to a Markov Process. The state of link $l$ at time $t$ is given by
\begin{align}
    l(t) = \begin{cases}
        1 & \text{Link $l$ is available for use at time slot } t\\
        0 & \text{otherwise}.
    \end{cases}
\end{align}
In other words, link $l$ is ON at time $t$ if $l(t) = 1$ and it is OFF at time $t$ if $l(t) = 0$.
\begin{figure}
    \centering
    \includegraphics[width=0.7\linewidth]{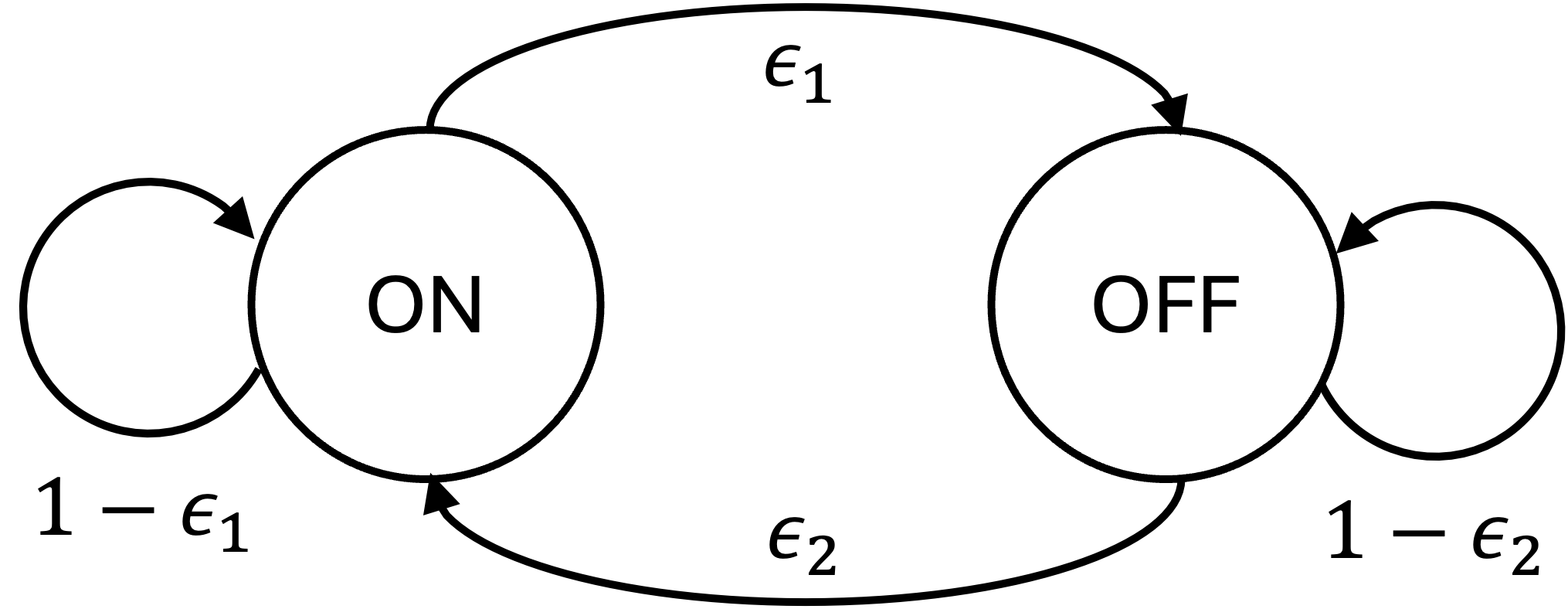}
    \caption{Markovian Link dynamics of $l(t)$}
    \label{fig:markov}
\end{figure}
 The distribution of $l(t)$ follows Markovian dynamics, as depicted by Figure \ref{fig:markov}. In particular,
\begin{equation}\label{eqn:markov_model}
\begin{aligned}
    \P(l(t) = 1 \;|\; l(t-1) = 0) &= \eps_2\\
    \P(l(t) = 0 \;|\; l(t-1) = 0) &= 1-\eps_2\\
    \P(l(t) = 1 \;|\; l(t-1) = 1) &= 1-\eps_1\\
    \P(l(t) = 0 \;|\; l(t-1) = 1) &=\eps_1
\end{aligned}.
\end{equation}
We assume that the state dynamics of each link are independent of the other for simplicity. However, even with this assumption, we can still model many phenomena over satellite links. 
For example, a link state $l(t)$ can model whether a queue buffer between two nodes is at capacity or not at capacity, hence capturing whether a link is congested or not. A second example would be when $l(t)$ models random optical link failures over time, where a failed optical link will take some random amount of time to be set up again and a functional optical link will remain functional for some random duration. Another example is when the satellite links are allocated based on user priority, where the link capacities are mainly reserved for the primary users and the unused capacity is allocated to secondary users \cite{handley_delay_2018}. The dynamics of the availability of capacity to secondary users could be captured by Markovian dynamics.

The link state transition probabilities $\eps_1$ and $\eps_2$ model how fast the link state changes compared to the time it takes to transmit a fixed-size packet from one node to the other (time slot). 
The steady-state probability of the link being ON is $p$ and it is given by,
\begin{align}
    p := \P(l(t) = 1) = \frac{\eps_2}{\eps_1 + \eps_2}.
\end{align}
Consequently, the probability of the link being OFF is $(1-p)$. We will also be interested in the $k-$stage transition probabilities of the links, namely,
\begin{align}
    \P\left(l(t+k) = j \;\bigg|\; l(t)=i \right) = p_{ij}(k).
\end{align}
We have the following expression for $p_{ij}(k)$,
\begin{equation}\label{eqn:transition_k}
\begin{aligned}
    p_{00}(k) &= \frac{\eps_1 + \eps_2(1-\eps_1-\eps_2)^k}{\eps_1 + \eps_2} = (1-p) + p\mu^k\\ 
    p_{01}(k) &= \frac{\eps_2 - \eps_2(1-\eps_1-\eps_2)^k}{\eps_1 + \eps_2}= p - p\mu^k \\
    p_{10}(k) &= \frac{\eps_1 - \eps_1(1-\eps_1-\eps_2)^k}{\eps_1 + \eps_2}= (1-p) -(1-p)\mu^k\\
    p_{11}(k) &= \frac{\eps_2 + \eps_1(1-\eps_1-\eps_2)^k}{\eps_1 + \eps_2}= p + (1-p)\mu^k,
\end{aligned}
\end{equation}
where $\mu := 1 - \eps_1 - \eps_2$ is the \textit{memory parameter} of the link. When the \textit{memory parameter} $\mu = 0$, the links have no memory (i.i.d. in time), and when $\mu \to 1$, the links will be static once realized. We assume that $\mu \geq 0$, i.e., the links have \textit{positive memory }\cite{johnston_wireless_2018}. This assumption ensures that a link that is ON is more likely to be ON after $k$ time slots, compared to the probability of a link that is OFF turning ON, i.e., $p_{11}(k) > p_{01}(l), \forall l$. This assumption allows for a controller with delayed network state information to make sensible routing decisions.

We are interested in routing between a source and a destination on a packet-by-packet basis. In a particular snapshot of the network, the satellite network forms a toroidal mesh grid (Figure \ref{fig:network-model}) over the Earth. The source and destination on the ground are generally served by the satellites closest to them, within their line of sight. To this end, routing a packet from the on-ground source to the on-ground destination in a particular snapshot essentially boils down to routing the message from the satellite node serving the source to the satellite node serving the destination over the dynamic satellite network. This is reasonable as it is fairly straightforward to precompute the satellite that would serve a geographic region (consequently, the on-ground source and destiantion) at any given time.

Hence, without loss of generality, we can fix the destination satellite node to have label $(0,0)$. Then, a satellite node $(x,y)$ is located $|x| + |y|$ hops away from the destination. It would take at least $|x| + |y|$ time slots to transmit a packet from $(x,y)$ to $(0,0)$. 
Without loss in generality, we  assume that $x > 0, y> 0$ and $|x| + |y| = x+y$ in the rest of the paper. This is due to the symmetry of the network along the vertical and horizontal axes.


\subsection{Routing Schemes}\label{sec:routing-schemes}
1) \textit{Centralized Routing Scheme}: 
In this scheme, we assume that a centralized router has delayed knowledge of the satellite links' state. This delay arises due to various factors, but it is primarily due to the time it takes for an electromagnetic wave to propagate from the satellite back to the Earth. Whenever a ground station (GS) needs to transmit a fixed-length packet to another GS through the satellite network, the centralized router finds a path for the packet using the network state information and appends the path onto the packet. Essentially, a \textit{virtual circuit} is established from the source GS to the destination GS through the satellite network. This form of routing is also sometimes referred to as \textit{virtual circuit routing} or \textit{layer 2 routing}.
Once the centralized router computes a path for the packet, the source GS forwards it to the first satellite (source satellite). One by one, the satellites decode the path appended with the packet, and forward it to the next satellite in the path, until the destination satellite is reached. The satellites have no onboard intelligence and only forward the packet to the next node according to the path appended in the header. The centralized router finds a path from source to destination with high probability, using stale global network state information. However, the links can change state after some time, due to their Markovian dynamics.

2) \textit{Distributed Routing Scheme}: In this scheme, every satellite has onboard routing intelligence and there is no centralized router. Every satellite knows the state of adjacent inter-satellite links in near real-time. However, the satellite router does not have global knowledge of all link states in real-time. Though it is possible to garner global network state information and use it for routing decisions, it leads to increased computation and communication complexity and is best avoided onboard satellites. This means that the router must make routing decisions with partial, yet real-time network state information.
The source ground station (GS) transmits its packet to the first satellite, and each satellite then decides the next satellite to which the packet must be forwarded in the network, until the destination GS is reached. This form of routing is also sometimes referred to as \textit{IP routing} or \textit{layer 3 routing}, as it involves route computation at each node.
\begin{figure}[!hbt]
    \centering
    \includegraphics[width=0.8\linewidth]{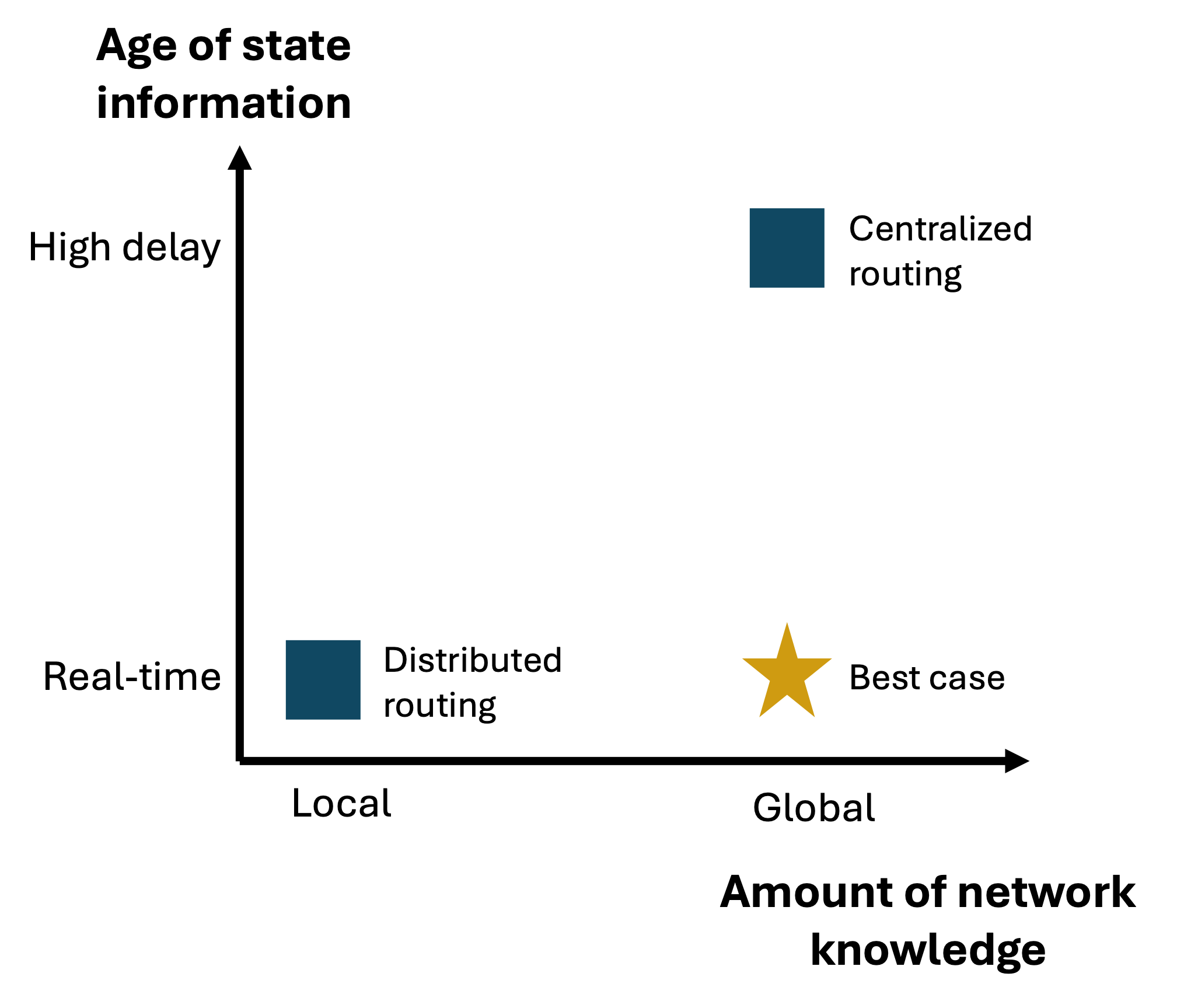}
    \caption{Tradeoffs of the Routing Schemes}
\label{fig:routing_tradeoff}
\end{figure}

Figure \ref{fig:routing_tradeoff} provides a high-level visual summary of the tradeoffs between the two routing schemes.

\subsection{Performance Metrics and Scenarios of Interest}\label{sec:scenarios-of-interest}
We are interested in two metrics, \textit{throughput} and \textit{delay}. We define throughput $T_{xy}$ as the probability of successfully routing a packet from the source $(x,y)$ to the destination $(0,0)$. 
We define the delay $D_{xy}$ as the average number of slots it takes to route a packet from source $(x,y)$ to destination $(0,0)$.
We compare the throughput and delay of the routing schemes in the network model defined above in two scenarios, 1) without onboard buffering and 2) with onboard buffering.

In the first scenario, the satellite nodes have no buffering capabilities. In other words, the packet in consideration is dropped by a satellite node whenever the respective outgoing link along the chosen path is in the OFF state. This means that the packet suffers no significant buffering or queueing delay along its path. Since the packet only experiences propagation and transmission delays, any routing policy would have nearly minimum possible delay in this scenario. However, the throughput will naturally suffer as it is likely that a packet will be dropped in transit. 
We are interested in routing policies that maximize throughput in this scenario.

In the second scenario, the satellite nodes have buffering capabilities. Consequently, the throughput is nearly maximum for any routing policy, as packets are rarely dropped. However, they can be significantly delayed due to buffering, while the packet waits for a link to become available again.  
We are interested in routing strategies that minimize the buffering delay in this scenario. 


There is an evident tradeoff between throughput and delay as visualized in Figure \ref{fig:throughput_delay}.
\begin{figure}
    \centering
\includegraphics[width=\linewidth]{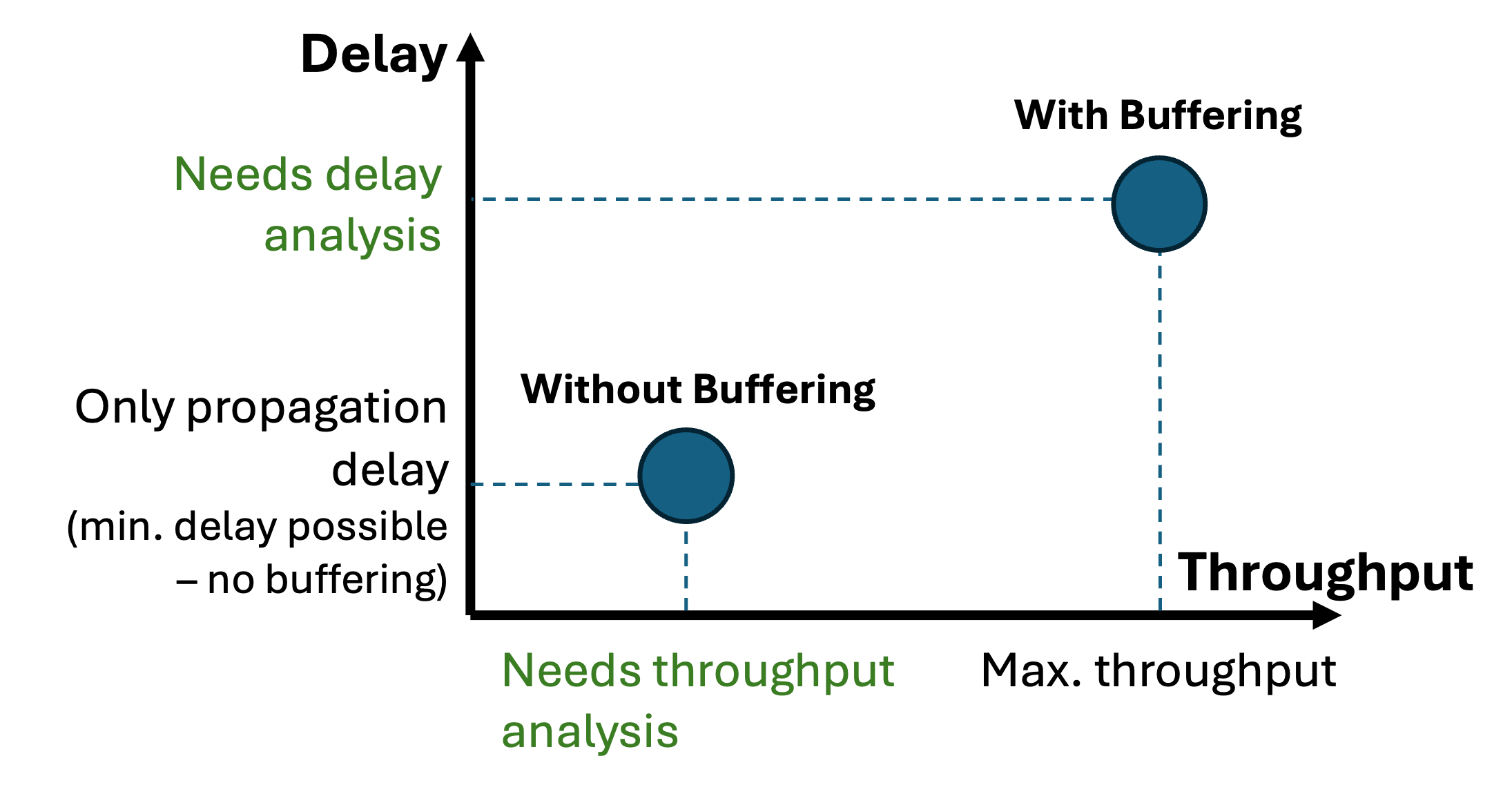}
    \caption{Pictoral representation of the two scenarios of interest.}
    \label{fig:throughput_delay}
\end{figure}
In the following sections, we analyze centralized and distributed routing schemes in the two scenarios.
\section{Analysis of a Centralized Routing Scheme}\label{sec:centralized_routing_analysis}
In this section, we introduce and motivate a simple centralized routing policy called the Shortest Connected Path Routing (SCPR). We will then analyze its performance without buffering and with buffering. 

Recall that in a centralized setting, the router gets global network state information after some inherent delay. 
Let parameter $t_c$ denote the amount of time since the network state information was most recently observed, until the packet reaches the source satellite. The parameter $t_c$ includes the propagation delays, path computation delay, as well as other delays in the system.

We will use $\gamma$ to denote a path from $(x,y)$ to $(0,0)$. The path $\gamma$ is a collection of links which connects $(x,y)$ to $(0,0)$. Now, we define the term ``connected" in the Shortest Connected Path Routing.
\begin{definition}[Connected Path]
A path $\gamma$ is \textit{connected} at time $t$ if all the links $l \in \gamma$ are ON at time $t$. In other words, $l(t) = 1$, $\forall l \in \gamma$. 
\end{definition}
The Shortest Connected Path Routing (SCPR) is a simple practical algorithm. The policy observes the network state at the current time $t$ and selects the shortest connected path from the source $(x,y)$ to the destination $(0,0)$. If there exists no connected path between $(x,y)$ and $(0,0)$, then it just selects one of the shortest paths at random. SCPR is visually summarized in Figure \ref{fig:scpr_routing}. 

SCPR is representative of the class of centralized policies since any such policy would have to make routing decisions on delayed network state information. 
As discussed in the introduction, the optimal centralized policy to maximize throughput or minimize delay requires solving a complex dynamic program, and SCPR is not always optimal.
However, it is possible to show that the shortest \textit{connected} path is throughput optimal if it also happens to be one of the shortest paths from $(x,y)$ to $(0,0)$. In fact, the shortest connected path between is also very likely to be a shortest path based on the numerical analysis in \cite{zhang_connectivity_2013}.


\begin{figure}
    \centering
\includegraphics[width=0.7\linewidth]{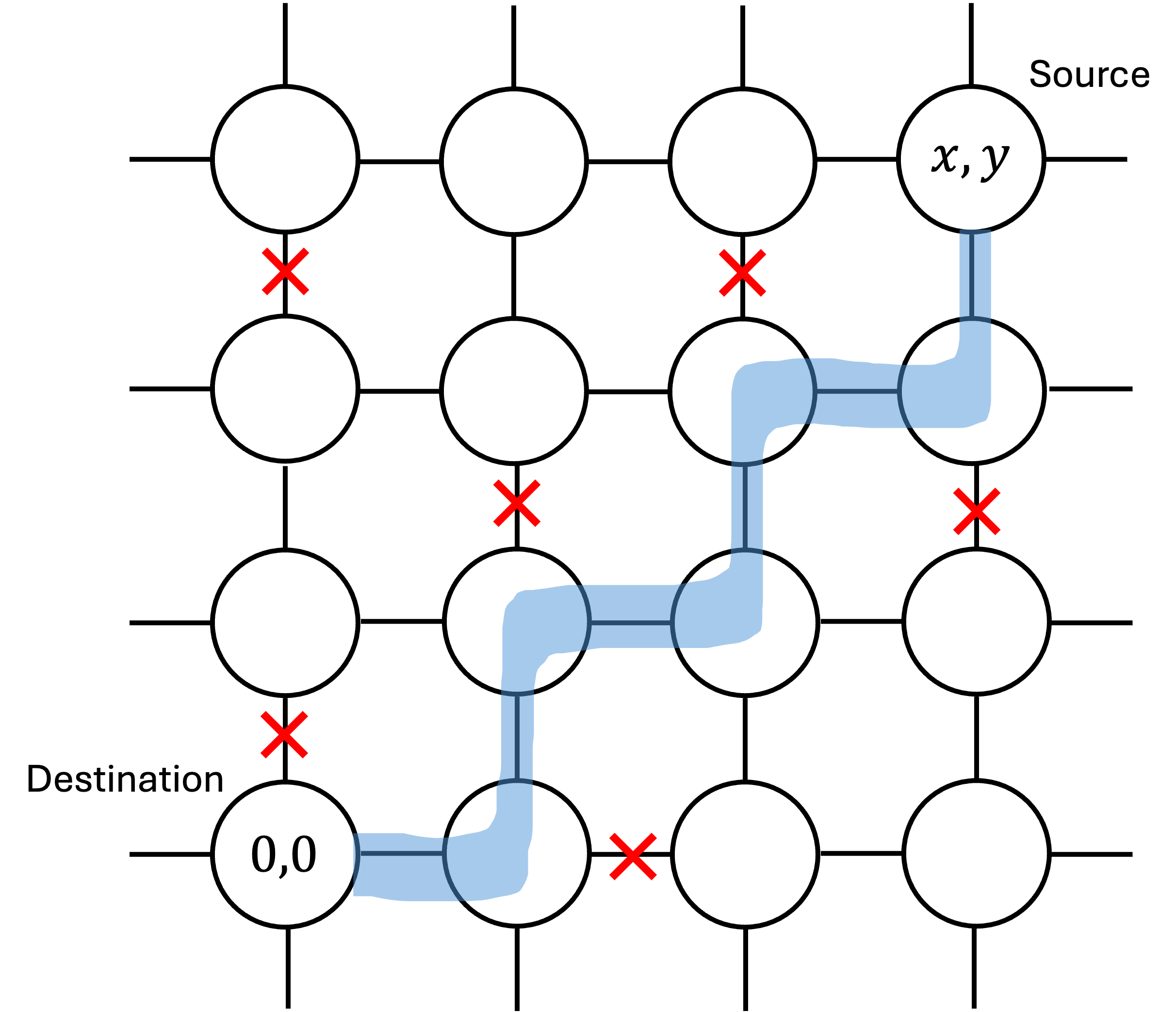}
    \caption{Given the network state at $t=0$, the SCPR policy gives us the shortest connected path from source to destination, as highlighted. Here, the red crosses denote links being unavailable for use.}
    \label{fig:scpr_routing}
\end{figure}

In the following  subsection,
we will analyze SCPR throughput in the case where satellites do not have onboard buffering capabilities.

\subsection{Shortest Connected Path Routing Without Buffering}\label{sec:centralized_routing_in_scenario_A}

The throughput for the SCPR policy in this scenario is given by the probability that there exists a path to the destination and the probability that all links are ON at the time they are traversed. Claim \ref{claim:centralized_throughput} states an upper bound for throughput under the SCPR policy. 

\begin{claim}[Throughput Under SCPR Policy]\label{claim:centralized_throughput} For the centralized SCPR policy, the throughput is upper bounded as
    \begin{equation}
    T^{SCPR}_{xy}\leq \prod_{i=0}^{x+y- 1} \left(p + (1-p)\mu^{t_c + i}\right).\label{eqn:scpr_throughput_approx}
\end{equation}
\end{claim}
We provide an upper bound because exactly characterizing  $T^{SCPR}_{xy}$ requires counting all cuts from source to destination in the graph which is provably NP-hard \cite{provan_complexity_1983}.
However, the upper bound is a good approximation to $T^{SCPR}_{xy}$, and is close to its true value as seen in Section \ref{sec:comparison}. It is interesting to note that the throughput decays exponentially with the distance between the source and destination, as well as the inherent delay, $t_c$ (for fixed $\mu$). 

\begin{proof}[Proof of Claim \ref{claim:centralized_throughput}]
Let the central router get network state information at $t = 0$ without loss in generality. The source satellite receives the packet from the source GS at time $t = t_c$, after the unavoidable delay $t_c$. 
Let $\gamma$ be the shortest connected path given the network configuration at $t=0$. Then, the throughput $T^{SCPR}_{xy}$ is the probability that every link in $\gamma$ remains ON when the packet reaches it. The probability that the $i^{th}$ link is ON when reached is $p_{11}(t_c + i)$ (as defined in \eqref{eqn:transition_k}).  Therefore,
\begin{align}\label{eqn:success_p}
    T^{SCPR}_{xy}= \sum_{j={x+y}}^\infty \prod_{i=0}^{j-1} p_{11}(t_c + i)\P(|\gamma| = j).
\end{align}
Here, $\P(|\gamma| = j)$ denotes the probability that the length of the shortest connected path from $(0,0)$ to $(x,y)$ is $j$, and $\P(|\gamma| = \infty)$ denotes the probability that there exists no connected path to the destination.

It is important to point that $\P(|\gamma| = j)$ is difficult to compute for large $N$ and $M$ because it would entail counting all possible graph cuts between the source and the destination, which is known to be NP-hard \cite{provan_complexity_1983}. Reference \cite{zhang_connectivity_2013} provides a numerical algorithm to compute the probability that $(x,y)$ is connected to $(0,0)$ through a shortest path in the grid.  Nevertheless, numerical analysis in  \cite{zhang_connectivity_2013} suggests that $(x,y)$ is connected to $(0,0)$ through a shortest path (of length $x+y$) with high probability for large values of $p$.  

However, we can compute a simple upper bound for the throughput $T^{SCPR}_{xy}$. Since $j \geq x + y$ and $p_{11}(\cdot) \leq 1$, we get, 
\begin{align*}
    T^{SCPR}_{xy} &\leq \sum_{j=x+y}^\infty \prod_{i=0}^{x+y- 1} p_{11}(t_c + i) \P(|\gamma| = j) \nonumber\\
    &= \prod_{i=0}^{x+y- 1} p_{11}(t_c + i) \underbrace{\sum_{j = x + y}^{\infty} \P(|\gamma| = j)}_{=1}\nonumber\\
    &=  \prod_{i=0}^{x+y- 1} \left(p + (1-p)\mu^{t_c + i}\right). 
\end{align*}
\end{proof}

\subsection{Shortest Connected Path Routing With Buffering}\label{sec:centralized_routing_in_scenario_B}
If we define $X_i$ as the delay incurred to traverse the $i^{th}$ link in $\gamma$, then, the total expected delay due to SCPR, $D^{SCPR}_{xy}$ is 
\begin{equation*}
    D^{SCPR}_{xy} = \E\left[\sum_{i=0}^{|\gamma|-1}X_i\right]. 
\end{equation*}
It is important to note that the distribution of $X_i$ depends on the delay incurred before reaching link $i$, i.e., $\sum_{k=0}^{i-1}X_k$. This makes the expression for the delay $D^{SCPR}_{xy}$ non-trivial to compute. Claim \ref{claim:centralized_delay} provides 
a recursive method to lower bound the mean delay of the SCPR policy.

\begin{claim}\label{claim:centralized_delay} For the centralized SCPR policy with buffers, the mean delay is lower bounded as
    \begin{equation}\label{eqn:centralized_delay}
    D^{SCPR}_{xy} \geq \frac{d}{dt}M_{x+y}(t)\bigg|_{t=0}
    \end{equation}
where, $$M_i(t)\triangleq \E\left[\mu^{t\sum_{k=0}^{i-1}X_k}\right]/\log \mu, $$  
and it satisfies the following recursive equation
$$M_i(t) = A(t)M_{i-1}(t) + B(t)M_{i-1}(t+1),$$
and
\begin{align*}
     A(t) &= \frac{p(1 - \mu^t)\mu^t + \eps_2\mu^{2t}}{1-(1-\eps_2)\mu^t}\\
     B(t) &= \frac{(1-p)\mu^{t_c + t}(1-\mu^t)}{1-(1-\eps_2)\mu^t}\\
     M_0(t) &= 1/\log(\mu).
 \end{align*}
\end{claim}

Even though Claim \ref{claim:centralized_delay} provides a lower bound for the delay, it is a close approximation for the true mean delay of the SCPR policy for large values of $p$. This is also verified in Section \ref{sec:comparison} through simulation. The approximate value can be efficiently numerically computed using the recursive relation in Claim \ref{claim:centralized_delay}. 
We now provide the proof.
\begin{proof}[Proof of Claim \ref{claim:centralized_delay}]
Firstly, note that even when there exists no connected path between the source and destination, the mean delay along any of the shortest paths to the destination is finite as the every link on the path will be ON eventually and infinitely often with probability 1.
Using the lower bound $|\gamma| \geq x + y$,
one can show that 
\begin{equation*}
    D^{SCPR}_{xy} = \E\left[\sum_{i=0}^{|\gamma| - 2}X_i\right] \geq \E\left[\sum_{i=0}^{x+y-1}X_i\right].
\end{equation*}
When a packet reaches a satellite node and the next link along the path is OFF, then the node buffers the packet until the link becomes ON and available to use. If the link is OFF, the delay $X_i$ experienced by the packet along link $i$ is a geometric random variable (c.f. the Markov model as described in \eqref{eqn:markov_model}). The probability of link $i$ being ON or OFF clearly depends on the aggregate delay incurred before reaching link $i$, which is $t_c + \sum_{k=0}^{i-1}X_k$. 
In particular, $X_i$ behaves as follows:
\begin{align}
    X_i = \begin{cases}
        1 & \text{ if the $i^{th}$ link is ON when reached}\\
        1 + Z_i & \text{ if the $i^{th}$ link is OFF when reached}.
    \end{cases}
\end{align}
Note that $Z_i \sim Geometric(\eps_2)$ as discussed earlier. 
The packet reaches the $i^{th}$ link along the path after a delay of $S_i \triangleq \sum_{k=0}^{i-1}X_k$ from the source satellite. From this, we get
{\medmuskip=0mu
\begin{align}
    \P(\text{Link $i$ is ON when reached } | S_i) &= p_{11}(t_c + S_i)\nonumber\\
    &= p + (1-p)\mu^{tc + S_i}\\
    \P(\text{Link $i$ is OFF when reached }| S_i) &= p_{10}(t_c + S_i)\nonumber\\
    &= (1-p) - (1-p)\mu^{tc + S_i}.
\end{align}
}

At this stage, it is convenient to use the Moment Generating Function of $S_i$ to compute $\E[S_i]$. We define the Moment Generating Function (MGF) as follows
\begin{equation}
    M_{i}(t) \triangleq \E[\mu^{tS_i}]/\log\mu. \label{eqn:mgf}
\end{equation}
Using the MGF, we can compute all the moments of $S_i$ using the derivatives of $M_i(t)$. In particular, we are interested in the mean, and
\begin{align}
    \E[S_i] = \frac{d}{dt}M_i(t)\bigg|_{t=0}.\label{eqn:mean_and_mgf}
\end{align}
Then,
\begin{align*}
    \E\left[\mu^{tS_i}\right] &= \E\left[\E\left[\mu^{tS_i}|S_{i-1}\right]\right]\\
    &= \E\left[\E\left[\mu^{t(X_i + S_{i-1})}|S_{i-1}\right]\right]\\
    &= \E\left[\mu^{tS_{i-1}}\E\left[\mu^{tX_i}|S_{i-1}\right]\right]
\end{align*}
From the distribution of $X_i$ conditioned on $S_{i-1}$, it follows that
\begin{align}
   \E\left[\mu^{tX_i}|S_{i-1}\right] &= \mu^t(p + (1-p)\mu^{tc + S_i}) \nonumber\\
   &\quad+ \mu^t\E[\mu^{tZ_i}](1-p -(1-p)\mu^{tc + S_i}) \label{eqn:cond_mgf}.
\end{align}
The moment generating function for a geometric random variable is known to be as follows.
\begin{align*}
   \E\left[\mu^{tZ_i}\right] &= \frac{\eps_2\mu^{t}}{1 - (1-\eps_2)\mu^t}.
\end{align*}
On simplifying \eqref{eqn:cond_mgf}, we get, 
\begin{align*}
    \E\left[\mu^{tX_i}|S_{i-1}\right] &= \frac{p(1 - \mu^t)\mu^t + \eps_2\mu^{2t}}{1-(1-\eps_2)\mu^t}\mu^{tS_{i-1}} \nonumber\\
    &\quad + \frac{(1-p)\mu^{t_c + t}(1-\mu^t)}{1-(1-\eps_2)\mu^t}\mu^{(t+1)S_{i-1}}.
\end{align*}
On taking an expectation of the above expression, we get a recursive relation for $M_i(t)$. 
\begin{equation}
    M_{i}(t) = A(t)M_{i-1}(t) + B(t)M_{i-1}(t+1),
 \end{equation}
 where, 
 \begin{align}
     A(t) &= \frac{p(1 - \mu^t)\mu^t + \eps_2\mu^{2t}}{1-(1-\eps_2)\mu^t}\\
     B(t) &= \frac{(1-p)\mu^{t_c + t}(1-\mu^t)}{1-(1-\eps_2)\mu^t}.
 \end{align}
 Finally, $$D_{xy}^{SCPR} \geq \E\left[\sum_{i=0}^{x+y-1}X_i\right] = \frac{d}{dt}M_{x+y}(t)\bigg|_{t=0}.$$


\end{proof}

\section{Analysis of a Distributed Routing Scheme}\label{sec:distributed_routing_analysis}
In this section, we consider the setting where satellites have onboard routing capabilities. In this setting, the satellite nodes know the state of nearby links (local information) in real-time and make real-time routing decisions.  A practical shortest-path distributed routing policy is the greedy routing (GR) policy. The principle of the GR policy is to move closer to the destination at every step, whenever possible. It makes sense to do so as we have only restricted information about the network state. Clearly, the GR policy always takes the shortest path to the destination.  To define the GR policy more formally, let the packet initially be at $(x,y)$ that $x, y > 0$ for the sake of simplicity. Moreover, let $l_h$ and $l_v$ be the links along the horizontal and vertical directions respectively, that move toward the destination, i.e., moving along $l_v$ at $(x,y)$ would take you to $(x,y-1)$, and moving along $l_h$ would take you to $(x-1,y)$. At every node $(x', y')$, repeat the process (see Figure \ref{fig:greedy_policy}). We define the GR policy as follows--
\begin{enumerate}[wide, labelwidth=!, labelindent=0pt]
    \item At some node $(x',y')$, observe the state of the local outgoing links $l_v$ and $l_h$.
    \item One of the following cases occurs--
    \begin{enumerate}
        \item $x' > 0$ and $y' > 0$. We are in the interior of the grid between the source and the destination. In this case, we have two edges that can potentially take us closer to the destination, $l_v$ and $l_h$. We further have four cases.
        \begin{enumerate}
        \item \textit{$l_v$ and $l_h$ are both ON}: Choose to move along $l_v$ with probability $u$ (or along $l_h$ with probability $\bar u = 1-u$).
        \item \textit{$l_v$ is OFF, $l_h$ is ON}: Choose to move along $l_h$.
        \item \textit{$l_h$ is OFF, $l_v$ is ON}: Choose to move along $l_v$.
        \item \textit{$l_v$ and $l_h$ are both OFF}: Depending on whether or not we use onboard buffers, we either buffer the packet or drop it.
    \end{enumerate}
    \item $x' = 0$. In this case, we are at the boundary. Only the $l_h$ link can take us closer to the destination. We move along $l_h$ if it is ON. Else, we drop the packet or buffer it depending on whether the satellites have buffers.
    \item $y' = 0$. This case is similar to the previous case. We move along $l_v$ if it is ON. Else, we drop the packet or buffer it depending on whether the satellites have buffers.
    \end{enumerate}
    \item We repeat this policy until we reach the destination, i.e. $x'=0$ and $y'=0$.
\end{enumerate}
\begin{figure}
    \centering
\includegraphics[width=0.7\linewidth]{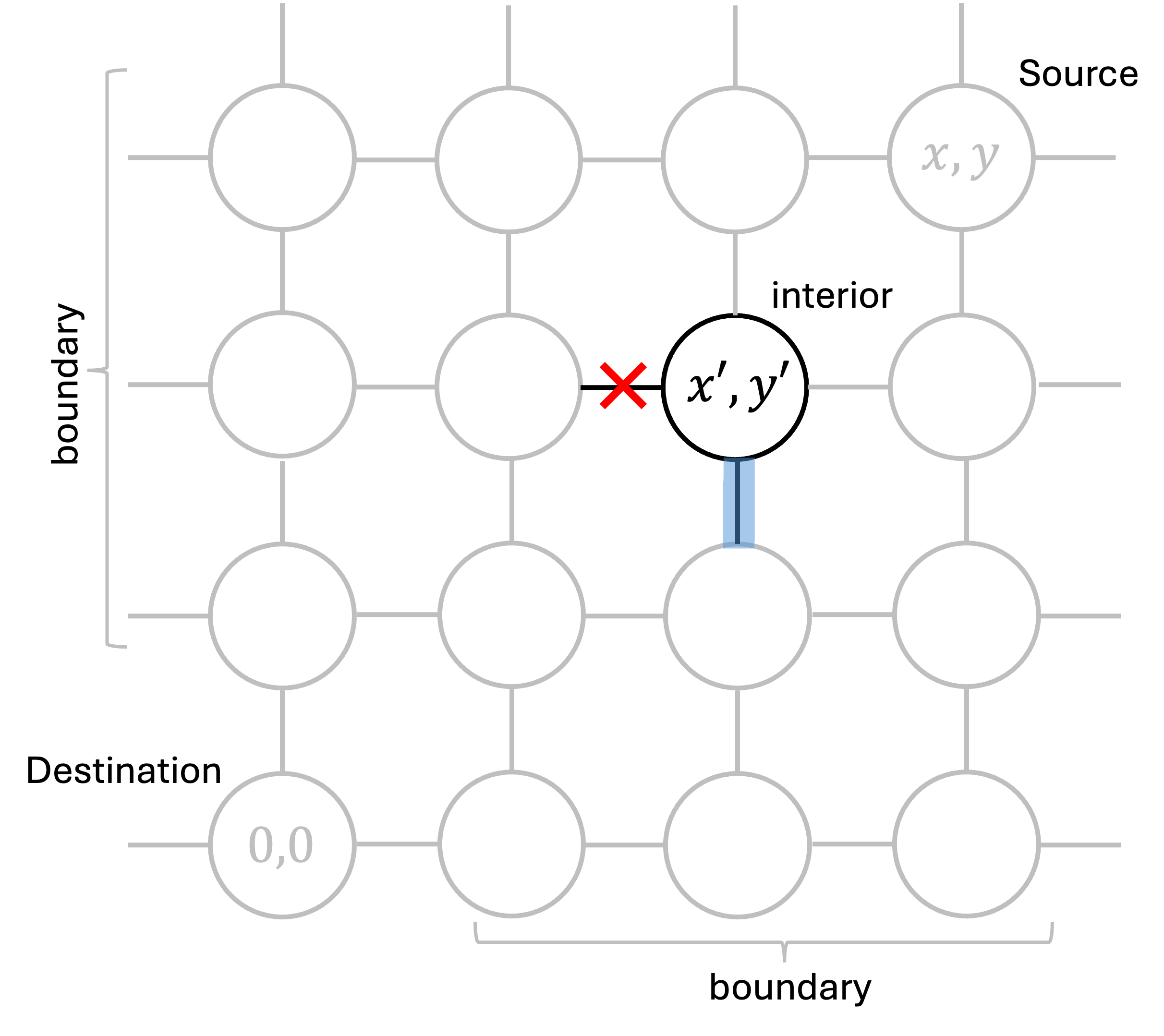}
    \caption{Greedy Routing (GR) Policy at node $(x',y')$. The nodes and links with a dark outline are the ones under consideration. The GR policy only considers the local links going towards the destination for routing. Since the horizontal link is unavailable, the routing decision is to move along the vertical, as highlighted.}
\label{fig:greedy_policy}
\end{figure}

It is important to point that the above algorithm ignores the memory in the network. Every routing decision is made only using the state of the outgoing links at a given node. Therefore, the state of the current outgoing links are \textbf{independent} of all routing decisions made earlier. 
This point is very important to the analysis that follows.

The GR policy is not the optimal distributed policy in general. For the general case, finding the optimal distributed routing problem becomes a complex dynamic program as discussed in the introduction. 
We study the GR policy as a representative distributed routing algorithm to simplify analysis. In the policy defined above, we use randomization to break ties when both links $l_h$ and $l_v$ are ON, i.e., we choose $l_v$ with probability $u$ and $l_h$ with probability $\bar u = 1-u$ to simplify analysis.
In the rest of the section, we analyze the greedy policy's performance.

\subsection{Greedy Routing (GR) Policy Without Buffering}\label{sec:greedy_routing_scenA}
When there are no buffers onboard satellites, we are interested in the throughput of the routing policy. Claim \ref{claim:distributed_throughput} states the expression for the throughput of the greedy policy.
\begin{claim}\label{claim:distributed_throughput}
    The throughput of the Greedy Routing policy is given as
    \begin{equation}\label{eqn:distributed_throughput}
        T^{GR}_{xy} = p^{y}\left(\frac{1-up}{\bar u}\right)^x I_{\bar u p}(x,y) + p^{x}\left(\frac{1-\bar up}{u}\right)^y I_{ u p}(y,x),
    \end{equation}
    where, $u$ is probability of routing along the vertical direction when both links are ON, $\bar u = 1 - u$ and  $I_v(a,b)$ denotes the Regularized Incomplete Beta Function \cite{feller_probability_1991} defined as
\begin{equation}
    I_v(a,b) = \frac{\int_{0}^{v}t^{a-1}(1-t)^{b-1}dt}{\int_{0}^{1}t^{a-1}(1-t)^{b-1}dt}.\label{eqn:I_v(a,b)}
\end{equation}
\end{claim}

In the rest of the subsection, we focus on the proof sketch of the above claim. To that end, let us define the following random variables:
\begin{itemize} 
\label{var_definitions}
    \item Let $\tau_x$ denote the first time that the packet successfully reaches $x'=0$.
    \item Let $\tau_y$ denote the first time that the packet successfully reaches $y'=0$.
    \item Let $E$ be the event that we successfully reach the destination. As the GR policy takes $x + y$ steps to reach the destination, event $E$ is equivalent to the event that there is always at least one outgoing link towards the destination for $x+y$ steps. 
\end{itemize}

We are primarily interested in understanding the distribution of $\min(\tau_x, \tau_y)$. This is because after $\min(\tau_x, \tau_y)$ the packet is routed through the boundary of the grid (where $x'=0$ or $y'=0$). This is important because on the \textbf{boundary} of the grid \textbf{there is only one shortest path to the destination}. If the only link which takes us toward the destination is OFF, we drop the packet. This is contrary to the case in the \textbf{interior}, where we always have \textbf{two potential links toward the destination}. Figure \ref{fig:greedy_policy} shows the interior and boundary regions for the GR policy. Therefore, we want to avoid reaching the boundary early so as to always have the option to take one of the two links that takes us closer to the destination. Geometrically, this means that we want to take a path to the destination that is as close to the diagonal between $(x,y)$ and $(0,0)$ as possible.

\begin{proof}[Proof of Claim \ref{claim:distributed_throughput}]
 Let us assume $y \geq x$ without loss of generality.
 Then note that,
\begin{equation}\label{eqn:prob_breakdown1}
    T^{GR}_{xy} = \P(E) = \sum_{k = x}^{x+y-1}\P\big(\{\min(\tau_x, \tau_y) = k\} \cap E\big),
\end{equation}
since it would take at least $x$ steps to reach the boundary, and no more than $x+y-1$ to reach the boundary.
Now, observe that $\tau_x \neq \tau_y$, since $(0,0)$ can only be reached through $(0,1)$ or $(1,0)$ meaning that $\tau_x > \tau_y$ or $\tau_x < \tau_y$. Therefore,
{
\medmuskip=2mu
\thinmuskip=2mu
\begin{align}\label{eqn:prob_breakdown2}
    &\P(\{\min(\tau_x, \tau_y) = k\} \cap E)\nonumber\\
    &= \begin{cases}
        \P(\{\tau_x > \tau_y = k\} \cap E) + \P(\{\tau_y > \tau_x = k\} \cap E) & k \geq y\\
        \P(\{\tau_y > \tau_x = k\} \cap E) & \text{else.}
    \end{cases}
\end{align}
}
The set $\{\tau_y > \tau_x = k \} \cap E$ describes the event the packet reaches the destination after reaching the boundary in $k$ slots. 
Since the link states at any step are independent of the routing decisions made before it, observe that the probability that we do not drop the packet is 
$$\underbrace{(1-(1-p)^2)^k}_{\parbox{8em}{\footnotesize Prob. that at least one link is ON at each node in the interior for $k$ slots}}\times \underbrace{p^{x+y-k}}_{\parbox{8em}{\footnotesize Prob. that link towards destination is ON for remaining slots along the boundary}} = p^{x+y}(2-p)^k.$$
Additionally, it is mandatory to move along the horizontal in the $k^{th}$ slot. The remaining $x-1$ horizontal steps happen in the previous $k-1$ slots. There are $\binom{k-1}{x-1}$ ways to move horizontally in $k-1$ slots. Conditioned that a move is made, the probability of moving vertically is $w = \frac{up^2 + p(1-p)}{(1-(1-p)^2)} = \frac{1 - (1-u)p}{2-p}$. Consequently, the probability of moving horizontally is $\bar w = 1-w$. 
Finally, after simplifying, it can be shown that,
{\medmuskip=1mu
\thinmuskip=1mu
\begin{align*}
    \P(\{\tau_y > \tau_x = k\} \cap E) &=
    \binom{k-1}{x-1} (1 - \bar up)^{k-x}(1 - up)^x p^{x+y}\\
    \P(\{\tau_x > \tau_y = k\} \cap E) &=
    \binom{k-1}{y-1}(1 - \bar up)^{y}(1 - up)^{k-y} p^{x+y}.
\end{align*}
}
Substituting these back into \eqref{eqn:prob_breakdown1} and \eqref{eqn:prob_breakdown2} and simplifying, we arrive at the following expression. 
\begin{align*}
    T^{GR}_{xy}
    =\;\; &p^{y}\left(\frac{1-up}{\bar u}\right)^x\;\; \underbrace{\sum_{k=0}^{y-1}\binom{k+x-1}{k}(1-\bar up)^{k}(\bar u p)^x}_{(I)}\\
    +\;\;&p^{x}\left(\frac{1-\bar up}{u}\right)^y\;\; \underbrace{\sum_{k=0}^{x-1}\binom{k+y-1}{k}(1-up)^{k}(u p)^y}_{(II)}.
\end{align*}
The combinatorial sum $(I)$ (and $(II)$) is in fact the probability of a Negative Binomial random variable with $x$ (and $y$ resp.) successes, with probability of success being $\bar u p$ (and $up$ resp.) being less than or equal to $y-1$ (and $x-1$ resp.). This probability can be expressed using the Regularized Incomplete Beta Function\cite{feller_probability_1991}, which gives the expression in Claim \ref{claim:distributed_throughput}.
\end{proof}

Observe that the throughput in Claim \ref{claim:distributed_throughput} is not symmetric with respect to $x,y$, and depends on $u$. But from the symmetry of the network, we would expect $T^{GR}_{xy} = T^{GR}_{yx}$.
Moreover, as mentioned earlier, we want to stay in the interior of the grid for as long as possible to increase the possibility of encountering an ON link toward the destination. Choosing $u = \frac{y}{x+y}$ (and $\bar u = \frac{x}{x+y}$) helps satisfy these requirements. Firstly, this would result in staying close to the diagonal. Secondly, \eqref{eqn:distributed_throughput} becomes symmetric as we pick $u$ appropriately depending on $x,y$.  One can also show that $u = \frac{y}{x+y}$ maximizes $T^{GR}_{xy}$, asymptotically as $x\to \infty, y \to \infty$, using large deviation theory.

 \subsection{Greedy Routing (GR) Policy With Buffering}\label{sec:greedy_routing_scenb}
 When the satellites have buffering capabilities, we are interested in the delay performance of the routing scheme. Claim \ref{claim:distributed_delay} gives an approximate formula for the mean delay of the GR Policy with buffers.
 \begin{claim}\label{claim:distributed_delay}
 The mean delay of the GR Policy is upper bounded as
\begin{equation}\label{eqn:greedy_delay}
\begin{aligned}
    D^{GR}_{xy} \leq \;\;\; &(x+y)\left(1 + \frac{(1-p)^2}{2\eps_2 - \eps_1^2}\right)\\
    + &\frac{(1-p)(1-\eps_2+p)}{2\eps_2-\eps^2_2}\sqrt{\frac{x + y}{2\pi w\bar{w}}},
\end{aligned}
\end{equation}
where, $w = \frac{y}{x+y}$ 
and $\bar w = 1-w$. 
\end{claim}

In the above expression, $w$ is actually the probability of routing along the vertical direction, and clearly, it depends on $u$, which is the probability of choosing the vertical direction to break ties. 


We now provide the proof of Claim \ref{claim:distributed_delay}. The analysis is similar to the one presented in \cite{basu_percolation_2012}, 
 where they analyze the asymptotic average delay per hop in a  square-grid store-and-forward network with i.i.d. links. Our analysis considers a non-asymptotic non-square grid setting where the links follow Markovian dynamics.

\begin{proof}[Proof of Claim \ref{claim:distributed_delay}]
Let us define $X_i$ to be the delay encountered by the packet in the $i^{th}$ step in the interior of the grid. Then,
\begin{equation}
    X_i = \begin{cases}
        1 & \text{w.p. }1-(1-p)^2\\
        1 + Z_i & \text{w.p. }(1-p)^2
    \end{cases}
\end{equation}
where $Z_i$ is a geometric random variable with parameter $1-(1-\eps_2)^2$, because the packet would have to wait until at least one outgoing link turns ON, out of the two. Let us define $X'_i$ to be the delay encountered by the packet in the $i^{th}$ step at the boundary. Then, 
\begin{equation}
    X'_i = \begin{cases}
        1 & \text{w.p. }p\\
        1 + Z'_i & \text{w.p. }1-p
    \end{cases}
\end{equation}
where $Z'_i$ is a geometric random variable with parameter $\eps_2$, because the packet would have to wait until the only outgoing link turns ON. We are interested in the mean delay to reach the destination, i.e.,

\begin{align}
    D^{GR}_{xy} = \E\left[\sum_{i=0}^{\min(\tau_x, \tau_y)-1}X_i + \sum_{i=\min(\tau_x, \tau_y)}^{x+y-1}X'_i\right],
\end{align}

where $\tau_x, \tau_y $ and $\min(\tau_x, \tau_y)$ are defined as in subsection \ref{var_definitions}. Since $\min(\tau_x, \tau_y)$ is a valid stopping time, we can apply Wald's Identity, to get the following expression.
\begin{align}
    D^{GR}_{xy} &= \E[X_i]\E[\min(\tau_x, \tau_y)]\nonumber\\ 
    &\;\;\;\;+ \E[X'_i](x+y- \E[\min(\tau_x, \tau_y)]).\nonumber
\end{align}
On substituting the values of $\E[X_i]$ and $\E[X'_i]$, we get
\begin{align}\label{eqn:dist_delay_1}
    D^{GR}_{xy}&= x+y + \frac{(1-p)^2}{2\eps_2 - \eps_2^2}\E[\min(\tau_x, \tau_y)]\nonumber\\ &\;\;\;\;+ \frac{1-p}{\eps_2}\big(x+y - \E[\min(\tau_x, \tau_y)]\big).
\end{align}
Since $\frac{(1-p)^2}{2\eps_2 - \eps_2^2} < \frac{1-p}{\eps_2}$, we want $\E[\min(\tau_x,\tau_y)]$ to be as close to $x+y$ as possible, implying that the \textbf{packet must be routed through the interior} of the grid.
Now, we need to compute $\E[\min(\tau_x, \tau_y)]$. To that end, we assume $y \geq x$ without loss of generality. Then,
{
\medmuskip=2mu
\thinmuskip=2mu
\begin{align}\label{eqn:prob_breakdown_delay_1}
    \P(\min&(\tau_x, \tau_y) = k) = \begin{cases}
        \P(\tau_y = k) + \P(\tau_x = k) & k \geq y\\
        \P(\tau_x = k) & \text{else.}
    \end{cases}
\end{align}
}
In the above expression,
\begin{align}
    \P(\tau_y = k) &= \binom{k-1}{y-1}w^{y}\bar w^{k-y}\\
     \P(\tau_x = k) &= \binom{k-1}{x-1}w^{k-x}\bar w^{x},
\end{align}
where $w$ is the probability that the packet moves vertically, and $\bar w = 1-w$ is the probability that the packet moves horizontally. The expression for $w$ depends on $u$, the probability of moving vertically in case of ties, and is given as
\begin{equation}
 w = up^2 + p(1-p) + (1-p)^2\left(u\frac{\eps_2}{2-\eps_2} + \frac{1-\eps_2}{2-\eps_2}\right).\label{eqn:w_wrt_u}   
\end{equation}
Thus, we have the following expression.  
\begin{align*}
\E[\min&(\tau_x,\tau_y)]\\
&= \sum_{k=x}^{x+y-1} k\cdot \P(\min(\tau_x, \tau_y) = k)\\
 &=\sum_{k=x}^{x+y-1} k \cdot \P(\tau_x = k) + \sum_{k=y}^{x+y-1} k \cdot\P(\tau_y = k)\\
&= x \cdot \sum_{k=x}^{x+y-1} \binom{k}{x}w^{k-x}\bar w^{x}+ y \cdot \sum_{k=y}^{x+y-1} \binom{k}{y}w^{y}\bar w^{k-y}.
\end{align*}
On massaging the above expression, we obtain,
\begin{align}
  \E[\min(\tau_x,\tau_y)] = \frac{x}{\bar w}  - \frac{w^{y-1}\bar w^{x-1}}{B(y,x)}+ \left(\frac{y}{w}-\frac{x}{\bar w}\right)I_{w}(x,y), \label{eqn:EK}  
\end{align}
where, $B(a,b)$ is the Beta function\cite{feller_probability_1991} and $I_v(a,b)$ is the Regularized Incomplete Beta function \cite{feller_probability_1991}.  
Further, we note that whenever 
\begin{align}\label{eqn:shape_condition}
    \frac{\min(x,y)}{y+x} \geq (1-p)\left(p + (1-p)\frac{1-\eps_2}{2-\eps_2}\right),
\end{align}
we can choose $u \in [0,1]$ in \eqref{eqn:w_wrt_u} appropriately so that $w = \frac{y}{x+y}$. We do this because setting $w=\frac{y}{x+y}$ implies that \textbf{the packet is routed along the diagonal on average}. On doing so,
\begin{align*}
    \E[\min(\tau_x, \tau_u)] &= (x+y) - \frac{y^{y-1}x^{x-1}}{(x+y)^{x+y-2}}\cdot \frac{\binom{x+y}{x}}{\frac{x+y}{xy}}\\
    &= (x+y) - \frac{y^{y}x^{x}}{(x+y)^{x+y-1}}\cdot \binom{x+y}{x}
\end{align*}
using the properties of the Beta function. Using Stirling's approximation as in \cite{gallager_info_1991}, i.e., $$\frac{1}{\sqrt{8}}\frac{(x+y)^{x+y+1/2}}{x^{x+1/2}y^{y+1/2}}\leq \binom{x+y}{x} \leq \frac{1}{\sqrt{2\pi}}\frac{(x+y)^{x+y+1/2}}{x^{x+1/2}y^{y+1/2}},$$
and on further simplifying, it follows that
\begin{equation}
    \E[\min(\tau_x, \tau_y)] \geq (x+y) - \sqrt{\frac{x+y}{2\pi w \bar w}}.
\end{equation}
Finally, we substitute the expression for $\E[\min(\tau_x, \tau_y)]$ in \eqref{eqn:dist_delay_1} to get the upper bound 
\begin{align*}
D^{GR}_{xy} \leq \;\;\; &(x+y)\left(1 + \frac{(1-p)^2}{2\eps_2 - \eps_1^2}\right)\\
    + &\frac{(1-p)(1-\eps_2+p)}{2\eps_2-\eps^2_2}\sqrt{\frac{x + y}{2\pi w\bar{w}}}.
\end{align*}

\end{proof}
Note that the expected delay scales linearly with the distance between the source and the destination. The last term in the delay is a correction term in the order of $\sqrt{x + y}$.

\section{Comparison Of Routing Strategies and Numerical Analysis}\label{sec:comparison}
Intuitively, a centralized routing policy would fare well when the links have a lot of memory, as the delayed network state information would still be valid after a long time, whereas, a distributed policy would work better in cases where there is no memory and the router needs to be adaptive. We confirm this intuition in this section.
First, we begin with a comparison in the case without buffers, and subsequently compare them in the case with buffers. Additionally, we plot the empirical performance of these algorithms in a $100 \times 100$ network using Monte-Carlo methods over 2000 trials, and show that the simulations match the analytical results.

\subsection{Comparison Of Routing Strategies Without Buffers}
In this subsection, we compare the throughput performance of the SCPR policy and GR policy without buffers. Comparing equations \eqref{eqn:scpr_throughput_approx} and \eqref{eqn:distributed_throughput}, we see that a sufficient condition for the GR policy to be better than the SCPR policy is
\begin{align}
    \left(\frac{1-up}{\bar up}\right)^x I_{\bar up}(x,y) + &\left(\frac{1-\bar up}{up}\right)^y I_{up}(y,x)\nonumber\\
    &\geq \prod_{i=0}^{x+y-1}\left(1 + \frac{1-p}{p}\mu^{t_c + i}\right).
\end{align}
We plot \eqref{eqn:scpr_throughput_approx} and \eqref{eqn:distributed_throughput} and the simulated throughput of the routing algorithms in Figure \ref{fig:throughput}. In the figures, note that the inherent delay $t_c$ is normalized to the time it takes to transmit a packet along one hop in the satellite network. In all the numerical analysis, we fix $p = 0.9$, $N=M=100$. In Figure \ref{fig:thro_mu}, observe that the SCPR policy only outperforms the GR policy when there is high memory in the system $\mu \approx 1$. This is because a connected path selected by the centralized controller is likely to be connected in the future only when there is high memory.  In Figure \ref{fig:thro_tc}, we fix $\mu = 0.99$ and study how SCPR's throughput varies with $t_c$. Observe that throughput of SCPR policy falls exponentially with increasing $t_c$. Even with high memory, GR policy outperforms the SCPR policy when $t_c > 35$. This is because the network is very likely to have changed state, rendering the planned path useless when the inherent delay of the centralized routing scheme is large. Lastly, we compare the throughput of these policies as a function of distance between the source and destination. In Figure \ref{fig:thro_x} we observe that the throughput of both policies falls drastically as the distances become large. In the small distance regime, with high memory and low inherent delay, a centralized policy can be better. However, for large distances, a distributed policy would dominate. 
{
\begin{figure*}[!hbt]
     \centering 
     \captionsetup{justification=centering}
     \begin{subfigure}{0.325\textwidth}
         \centering\includegraphics[width=\textwidth]{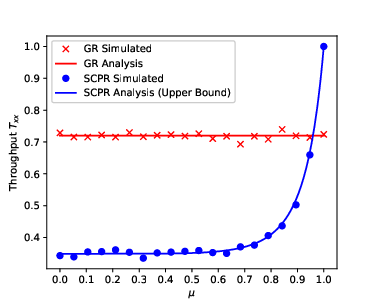}
         \caption{Throughput vs. $\mu$.\\We fix $p = 0.9$, $x=y=5$, $t_c=5$.}
         \label{fig:thro_mu}
     \end{subfigure}
     \hfill
     \begin{subfigure}{0.325\textwidth}
         \centering
         \includegraphics[width=\textwidth]{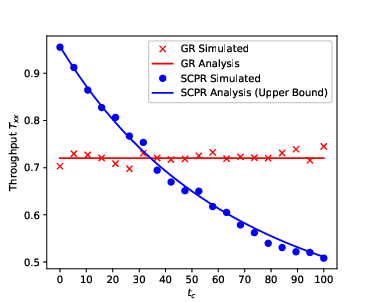}
         \caption{Throughput vs. $t_c$.\\We fix $p = 0.9$, $x=y=5$, $\mu=0.99$.}
         \label{fig:thro_tc}
     \end{subfigure}
     \hfill
     \begin{subfigure}{0.325\textwidth}
         \centering
         \includegraphics[width=\textwidth]{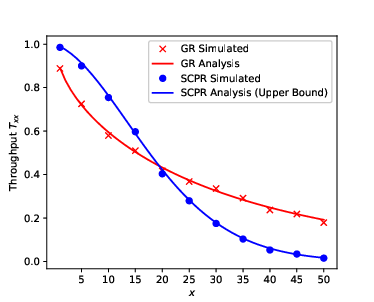}
         \caption{Throughput vs. $x$.\\We fix $p = 0.9$, $x=y$, $\mu=0.99$, $t_c=5$.}
         \label{fig:thro_x}
     \end{subfigure}
        \caption{SCPR and GR throughput against different network parameters}
        \label{fig:throughput}
\end{figure*}
\begin{figure*}[!hbt]
     \centering \captionsetup{justification=centering}
\begin{subfigure}{0.325\textwidth}
         \centering
\includegraphics[width=\textwidth]{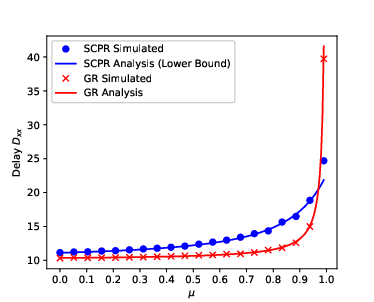}
         \caption{Delay vs. $\mu$.\\We fix $p = 0.9$, $x=y=5$, $t_c=5$.}
         \label{fig:del_mu}
     \end{subfigure}
    \hfill
     \begin{subfigure}{0.325\textwidth}
         \centering
         \includegraphics[width=\textwidth]{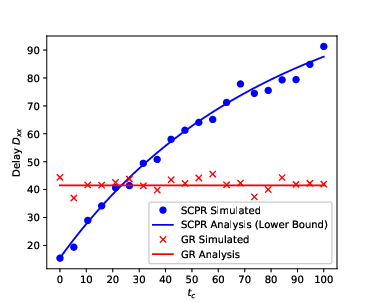}
         \caption{Delay vs. $t_c$.\\We fix $p = 0.9$, $x=y=5$, $\mu=0.99$.}
         \label{fig:del_tc}
     \end{subfigure}
     \hfill
     \begin{subfigure}{0.325\textwidth}
         \centering
    \includegraphics[width=\textwidth]{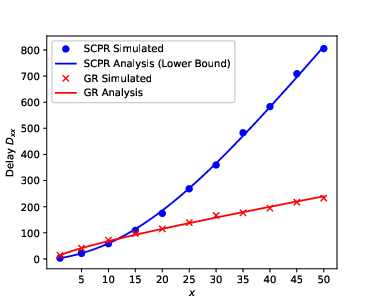}
         \caption{Delay vs. $x$.\\We fix $p = 0.9$, $t_c = 5$, $x=y$, $\mu=0.99$.}
         \label{fig:del_x}
     \end{subfigure}
        \caption{SCPR and GR delay against different network parameters}
        \label{fig:delay}

\end{figure*}
}
From the figures, we see that the simulated performance closely matches the analytical results. The mismatch is largely due to the variance of Monte-Carlo methods.

\subsection{Comparison of Routing Strategies With Buffers}
In this subsection, we compare the delay performance of the SCPR policy and GR policy with buffers. Comparing equations \eqref{eqn:centralized_delay} and \eqref{eqn:greedy_delay}, we see that a sufficient condition for the GR policy to be better than the SCPR policy is
\begin{align}
    (x+y)\left(1 + \frac{(1-p)^2}{2\eps_2 - \eps_2^2} + \frac{\frac{1-p}{\eps_2}\left(1-\frac{1-p}{2-\eps_2}\right)}{\sqrt{2\pi w\bar{w}(x+y)}}\right)\leq M'_{x+y}(0).  \label{eqn:del_suff_cond}
\end{align}
We plot the analytical mean delays, along with simulated delays, in Figure \ref{fig:delay} to better understand their behavior. We fix $p = 0.9$. In Figure \ref{fig:del_mu}, we plot the delay against $\mu$ and we observe that the delay increases with $\mu$ monotonically. Moreover, the GR policy greatly outpeforms the SCPR policy in terms of delay unless the network has high memory $\mu \approx 1$. In Figure \ref{fig:del_tc}, we see how the delay varies with increase in $t_c$ when $\mu = 0.99$. Note that the GR policy outperforms SCPR for $t_c > 32$. Similarly, looking at Figure \ref{fig:del_x} the GR policy is significantly better in terms of delay for large distances, and is comparable at small distances. 
From the figures, we verify that the simulation closely matches the analytical results. 



\section{Conclusions and General Remarks}\label{sec:conclusion}
As we see from the above numerical analysis, the centralized policy is only advantageous when there is high memory in the network links, $t_c$ is low and number of hops to reach the destination is small.
This happens when the time scale of state changes is much larger than the time it takes to transmit a packet across one hop, or $\eps_1 + \eps_2 \ll 1$. An example of this case would be when we want to model link unavailability due to optical link failures. The time it takes to establish an optical link is in the order of minutes which is much larger than the transmit time of a fixed size packet, and the links remain fairly stable once an optical connection is made. In this scenario, it would make sense to use a centralized scheme for route planning.
For other scenarios, it would make more sense to use a distributed policy.  
For example, if we want to model link unavailability due to user priority, the links transition from the ON to the OFF state and vice-versa, at around the same timescale as transmission of a packet along one hop. In this case, a distributed scheme makes more sense. 

We can also use these shortest path routing algorithms (like GR) as subroutines to construct modified policies that outperform the default shortest path routing algorithm (see Appendix \ref{appendix:improved_dist_routing}).

\section*{Acknowledgement}
We would like to thank SES S.A. for funding this research and also providing us with critical feedback in our research. We especially extend our gratitude to Joel Grotz and Valvanera Moreno for their valuable insights and critique throughout the course of this research.
\bibliographystyle{IEEEtran}
\bibliography{references}
\appendices



\section{Improved Distributed Routing}
\label{appendix:improved_dist_routing}
To improve throughput or delay of the default shortest path Greedy routing policy from node $(x,y)$ to $(0,0)$ when one dimension is much larger than the other, i.e. $|y| \gg |x|$ or $|x| \gg |y|$, we can forward the packet to an appropriate intermediate node $(u,v)$, which is more in the interior such that:
\begin{enumerate}
    \item For throughput: 
    $$
    T_{|x-u|,|y-v|}^{GR} \times T_{uv}^{GR} \geq T_{xy}^GR,
    $$
    \item For mean delay:
    $$
    D_{|x-u|,|y-v|}^{GR} + D_{uv}^GR \leq D_{xy}^GR,
    $$
\end{enumerate}
which improves performance. Here we use the GR policy as a subroutine to route the packet to the intermediate node $(u,v)$.

\end{document}